\documentclass[conference]{IEEEtran}
\IEEEoverridecommandlockouts
\usepackage{geometry}
\usepackage{graphicx}
\usepackage{amssymb}
\usepackage{amsmath}
\usepackage{cite}
\usepackage{subfigure}
\usepackage{mathrsfs}
\usepackage[displaymath,mathlines]{lineno}
\usepackage{color}
\usepackage{tabulary}
\usepackage{multirow}
\usepackage{algpseudocode}
\usepackage{algorithm,algpseudocode}
\usepackage{algorithmicx}
\usepackage{pbox}
\usepackage{multicol}
\usepackage{lipsum}
\usepackage{amsthm}
\usepackage{relsize}
\usepackage{lipsum}
\usepackage{epstopdf}
\usepackage{mathtools}
\usepackage{calc}

\makeatletter
\newcommand{\doublewidetilde}[1]{{%
		\mathpalette\double@widetilde{#1}}}
\newcommand{\double@widetilde}[2]{%
		\sbox\z@{$\m@th#1\widetilde{#2}$}%
		\ht\z@=.5\ht\z@
		\widetilde{\box\z@}}
\makeatother

\newtheorem{definition}{Definition}
\newtheorem{theorem}{Theorem}
\newtheorem{lemma}{Lemma}

\newtheorem{remark}{Remark}

\begin{document}
%
\title{\huge Massive MIMO under Double Scattering Channels: Power Minimization and Congestion Controls}

\author{
	 \IEEEauthorblockN{Trinh~Van~Chien$^{\ast}$, Hien~Quoc~Ngo$^{\dagger}$, Symeon~Chatzinotas$^{\ast}$, Bj\"{o}rn~Ottersten$^{\ast}$, and Merouane Debbah$^{\xi}$} 
	\IEEEauthorblockA{$^{\ast}$Interdisciplinary Centre for Security, Reliability and Trust (SnT), University of Luxembourg, Luxembourg\\
		$^{\dagger}$School of Electronics, Electrical Engineering and Computer Science, Queen's University Belfast, Belfast, UK \\
		$^{\xi}$CentraleSup\'elec, Universit\'e Paris-Saclay \&  Lagrange Mathematical and Computing Research Center, Paris, France 
	}
\thanks{The work of T. V. Chien, S. Chatzinotas, and  B. Ottersten was supported by FNR, Luxembourg
under the COREproject C16/IS/11306457/ELECTIC (Energy and CompLexity
EffiCienT mIllimeter-waveLarge-Array Communications). The work of H. Q. Ngo was supported by the UK Research and Innovation Future Leaders Fellowships under Grant MR/S017666/1.
}
}

\maketitle

\begin{abstract}
 This paper considers a massive MIMO system under the double scattering channels. We derive a closed-form expression of the uplink ergodic spectral efficiency (SE) by exploiting the maximum-ratio combining technique with imperfect channel state information. We then formulate and solve a total uplink data power optimization problem that aims at simultaneously satisfying the required SEs from all the users with limited power resources. We further propose algorithms to cope with the congestion issue appearing when at least one user is served by lower SE than requested. Numerical results illustrate the effectiveness of our proposed power optimization. More importantly, our proposed congestion-handling algorithms can guarantee the required SEs to many users under congestion, even when the SE requirement is high.
\end{abstract}

%
\IEEEpeerreviewmaketitle

\vspace*{-0.2cm}
\section{Introduction}
\vspace*{-0.15cm}
Wireless communications has sustained an exponential demand growth in data throughput over the last decades. However, mobile traffic will increase as foreseen in a short time with $12.3$ billion devices by $2022$ \cite{index2019global}. To handle this, massive MIMO, a disruptive technology, does not only inherit all the  multiplexing and diversity gains of the conventional MIMO but also offers large degree-of-freedoms as a consequence of equipping base stations (BSs) with many antennas \cite{massivemimobook}. Massive MIMO, therefore, provides unprecedented spectral and energy efficiency gains of modern wireless networks with only utilizing the contemporary time and frequency resources.

In massive MIMO, the closed-form SE expression can be obtained in certain scenarios such as rich scattering environments modeled by uncorrelated Rayleigh fading \cite{Ngo2013a} and references therein. Nonetheless, practical channels usually involve spatial correlation, which is modeled by correlated Rayleigh fading when the gathered energy at the antenna array comes from many directions likely leading to the full ranks of covariance matrices with an overwhelming probability \cite{Hoydis2013a,Chien2018a, yu2002models}. For rank deficiency occurring in poor scattering conditions, the authors in \cite{Gesbert2002a} proposed the \textit{double scattering channel} to characterize by the structure of scattering in the propagation environment and the spatial correlations around the transceiver. The first work numerically studying the uplink ergodic SE of cellular massive MIMO systems with the double scattering channels was found in \cite{van2016multi}. For theoretical analysis, the authors in \cite{kammoun2019asymptotic,Ye2020} computed the asymptotic ergodic SE of a single-cell massive MIMO system with the different linear precoding techniques when the number of BS antennas, scatterers, and users grow large with the same rate. To the best of our knowledge, no prior work analyzes the performance of massive MIMO systems with a finite number of BS antennas, users, and scatterers.

Many resource allocation tasks in massive MIMO communications can be implemented on the large-scale fading time scale \cite{Chien2017b}.  Notice that the key component of massive MIMO communications is allowing many users to access and share the radio resource at the same time with high quality of service. The max-min fairness optimization is  promising to provide uniform service to all the users \cite{massivemimobook}.  However, for large-scale networks with many base stations and users, the fairness level will approach a zero rate. In contrast, one can include separate SE constraints in the optimization problems to simultaneously maintain the quality of service for all the users \cite{van2020power}. Since users were randomly distributed, many user locations with poor channel conditions leads the optimization problems to be infeasible \cite{van2020uplink}.

By exploiting the double scattering channel model, this paper considers a massive MIMO system where a set of orthogonal pilot signals are reused by all the users. A new uplink ergodic SE expression is derived in closed form for a finite number of antennas at each base station (BS) and different  number of scatterers observed by every user and BS. After that,  we formulate and solve a total uplink data energy minimization problem subject to the required SE from every user and the power constraints. For user locations and shadow fading realizations, where our optimization problem is feasible, the global optimum can be obtained in polynomial time. We further propose two low computational complexity iterative algorithms that tackle the infeasible optimization problem by relaxing the SE constraints of unsatisfied users. Numerical results manifest the closed-form SE expression overlapping Monte-Carlo simulations. The effectiveness of the proposed data power control algorithms are compared with the interior-point methods.

\textit{Notation}: Upper-case/lower-case bold face letters are used to denote matrices and vectors, respectively. $\mathbf{I}_{M}$ is the identity matrix of size $M \times M$. $\mathbb{E} \{ \cdot \}$  denotes the expectation of a random variable. $\| \cdot \|$ is Euclidean norm. $\mathrm{tr}(\cdot)$ is the trace of a matrix. The regular and Hermitian transposes
are denoted by $(\cdot)^T$ and $(\cdot)^H$, respectively. Finally, $\mathcal{CN}(\cdot, \cdot)$ denotes the circularly symmetric complex Gaussian distribution.
\vspace*{-0.15cm}
\section{ Massive MIMO with Double Scattering Channels} \label{Sec:SystemModel}
\vspace*{-0.15cm}
We consider an uplink massive MIMO system comprising $L$~cells, where cell~$l$ has one BS equipped with $M$ antennas serving $K$ single-antenna users. A quasi-static channel model is used, where the time-frequency plane is divided into coherence blocks. Each coherence block has $\tau_c$ symbols for which the $\tau_p$ symbols are dedicated to the pilot training phase and the remaining $\tau_c - \tau_p$ symbols are used for the uplink data transmission (the downlink data transmission is neglected). The channel between user~$k$ in cell~$l$ and BS~$l'$ is modeled by \cite{van2016multi}, which is\footnote{This channel model was initiated for conventional MIMO systems under a far-field region and dedicated sub $6$-GHz bands for mobile services. In massive MIMO communications, the far-field effects are still observed since many antenna components can be installed in a small compact array \cite{bjornson2019massive}.}
\begin{equation} \label{eq:Channel}
\mathbf{h}_{lk}^{l'}  =  \sqrt{\beta_{lk}^{l'}/S_{lk}^{l'}} \left(\mathbf{R}_{lk}^{l'}\right)^{1/2} \mathbf{G}_{lk}^{l'} \left(\widetilde{\mathbf{R}}_{lk}^{l'} \right)^{1/2} \mathbf{g}_{lk}^{l'},
\end{equation}
where $\beta_{lk}^{l'}$ is the large-scale fading coefficient. $S_{lk}^{l'}$ is the number of scatterers generating the channel between BS~$l'$ and user~$k$ in cell~$l$. The matrix $\mathbf{R}_{lk}^{l'} \in \mathbb{C}^{M \times M }$ represents the correlation between the BS antennas and its scatterers; $\mathbf{G}_{lk}^{l'} \in \mathbb{C}^{M \times S_{lk}^{l'}}$ includes the corresponding small-scale fading coefficients. The matrix $\widetilde{\mathbf{R}}_{lk}^{l'} \in \mathbb{C}^{S_{lk}^{l'} \times S_{lk}^{l'} }$ is the correlation matrix between the transmit and receive
scatterers and $\mathbf{g}_{lk}^{l'} \in \mathbb{C}^{S_{lk}^{l'}}$ represents the corresponding small-scale fading coefficients. The elements of both $\mathbf{G}_{lk}^{l'}$ and $\mathbf{g}_{lk}^{l'}$ are independent and identically distributed as $\mathcal{CN}(0,1)$ conditioned on the trace of the covariance matrices. 
\vspace*{-0.15cm}
\subsection{Uplink Pilot Training}
\vspace*{-0.15cm}
In each coherence block, each BS needs instantaneous channel state information for the uplink data detection. The $\tau_p$ symbols are dedicated to the uplink pilot training, which can create $\tau_p$ mutually orthogonal pilot signals. User~$k$ in cell~$l$ uses the deterministic pilot signal $\pmb{\phi}_{lk} \in \mathbb{C}^{\tau_p}$ with $\| \pmb{\phi}_{lk}\|^2 = \tau_p$. This pilot signal is also reused by other users in multiple cells and we can define the pilot reuse set as
$\mathcal{P}_{lk} = \left\{ (l',k'): \pmb{\phi}_{l'k'} = \pmb{\phi}_{lk}, l = 1, \ldots, L, k'=1, \ldots, K \right\}$,
which contains the indices of all users sharing the same pilot signal as user~$k$ in cell~$l$, including $(l,k)$. Mathematically, it observes that $\pmb{\phi}_{lk}^{\rm H} \pmb{\phi}_{l'k'} =  \tau_p$ if  $(l',k') \in \mathcal{P}_{lk}$. Otherwise, $\pmb{\phi}_{lk}^{\rm H} \pmb{\phi}_{l'k'} = 0$. At BS~$l$, the received pilot signal $\mathbf{Y}_l^p \in \mathbb{C}^{M \times \tau_p}$ with the superscript $p$ standing for the pilot training phase is
\begin{equation}
\mathbf{Y}_l^p  = \sum_{l'=1}^L \sum_{k'=1}^{K} \sqrt{\hat{p}_{l'k'}} \mathbf{h}_{l'k'}^l \pmb{\phi}_{l'k'}^{\rm H} + \mathbf{N}_l^p,
\end{equation}
where $\mathbf{N}_l^p \in \mathbb{C}^{M \times \tau_p}$ is additive noise with the independent and identically random elements distributed as $\mathcal{CN}(0, \sigma^2)$. BS~$l$ estimates the channel $\mathbf{h}_{l'k'}^l$ from user~$k'$ in cell~$l'$ from
\begin{equation} \label{eq:ylklp}
\mathbf{y}_{l'k'}^{l,p} = \mathbf{Y}_{l}^p \pmb{\phi}_{l'k'} = \sum_{(l'',k'') \in \mathcal{P}_{l'k'}} \sqrt{\hat{p}_{l''k''}} \tau_p \mathbf{h}_{l''k''}^l + \mathbf{N}_l^p \pmb{\phi}_{l'k'}.
\end{equation}
The processed received signal $\mathbf{y}_{l'k'}^{l,p} \in \mathbb{C}^{M}$ has sufficient statistics to obtain a channel estimate of the origin $\mathbf{h}_{l'k'}^l$ by utilizing linear MMSE (LMMSE). 
\begin{lemma} \label{lemma:ChannelEstPhaseAware}
By utilizing the LMMSE estimation, the channel estimate $\hat{\mathbf{h}}_{l'k'}^l \in \mathbb{C}^{M}$ from user~$k'$ in cell~$l'$ and BS~$l$ is 
\begin{equation} \label{eq:ChannelEst}
\hat{\mathbf{h}}_{l'k'}^l = \sqrt{\hat{p}_{l'k'}} \beta_{l'k'}^l d_{l'k'}^l \mathbf{R}_{l'k'}^l \pmb{\Psi}_{l'k'}^l  \mathbf{y}_{l'k'}^{l,p},
\end{equation}
where $\pmb{\Psi}_{l'k'}^l = \left(  \sum_{(l'',k'') \in \mathcal{P}_{l'k'}} a_{l''k''}^l \mathbf{R}_{l''k''}^l + \sigma^2 \mathbf{I}_{M} \right)^{-1}$, with $a_{l''k''}^l =  \tau_p \hat{p}_{l''k''} \beta_{l''k''}^l d_{l''k''}^l$ and $d_{l'k'}^l = \mathrm{tr}\big(\widetilde{\mathbf{R}}_{l'k'}^l\big)/S_{l'k'}^l $. The covariance matrix of the channel estimate $\hat{\mathbf{h}}_{l'k'}^l$ is
\begin{equation} \label{eq:CovarianceEst}
\mathbb{E} \Big\{ \hat{\mathbf{h}}_{l'k'}^l \big(\hat{\mathbf{h}}_{l'k'}^l\big)^{\rm H} \Big\} =\hat{p}_{l'k'} \big(\beta_{l'k'}^l\big)^2 \big(d_{l'k'}^l \big)^2 \tau_p \mathbf{R}_{l'k'}^l \pmb{\Psi}_{l'k'}^l  \mathbf{R}_{l'k'}^l.
\end{equation}
\end{lemma}
\begin{proof}
The proof is based on the LMMSE estimation \cite{Kay1993a}, but adapted to our framework with the channel vector in \eqref{eq:Channel} and the pilot reuse pattern with non-Gaussian random variables. The detail proof is omitted due to space limitations.
\end{proof}
Lemma~\ref{lemma:ChannelEstPhaseAware} shows the concrete expression of the channel estimate of each user.  Our channel estimation considers the influence of pilot contamination in multi-cell massive MIMO scenarios, which is a generalization of the previous result in \cite{kammoun2019asymptotic,Ye2020} that assumed the orthogonal pilot signals for all the users in a single cell. 
\vspace*{-0.2cm}
\subsection{Uplink Data Transmission}
\vspace*{-0.15cm}
During the uplink data transmission, user~$k$ in cell~$l$ sends data symbol $s_{lk}$ with $\mathbb{E} \{ |s_{lk}|^2\} = 1$. The received data signal $\mathbf{y}_l \in \mathbb{C}^{M}$ at BS~$l$ is 
\begin{equation}
\mathbf{y}_l = \sum_{l'=1}^L \sum_{k'=1}^{K} \sqrt{p_{l'k'}} \mathbf{h}_{l'k'}^l s_{l'k'} + \mathbf{n}_l,
\end{equation}
where $p_{l'k'}$ is the transmit power of user~$k'$ in cell~$l'$ assigned to each data symbol and $\mathbf{n}_l$ is additive noise distributed as $\mathcal{CN}(\mathbf{0}, \sigma^2 \mathbf{I}_{M})$. By utilizing a combining vector $\mathbf{v}_{lk} \in \mathbb{C}^{M}$ and the use-and-then-forget channel capacity bounding technique as shown in \cite{massivemimobook}, the uplink ergodic SE is obtained as
\begin{equation} \label{eq:AchievableRate}
R_{lk} = \left(1 - \tau_p/\tau_c \right) \log_2 \left( 1 + \mathrm{SINR}_{lk} \right), [\mbox{b/s/Hz}],
\end{equation}
where the effective signal-to-interference-and-noise ratio (SINR) is given by \eqref{eq:GeneralSINR}, shown at the top of the next page.
\begin{figure*}
\begin{equation} \label{eq:GeneralSINR}
\mathrm{SINR}_{lk} = \frac{p_{lk} \big| \mathbb{E} \big\{ \mathbf{v}_{lk}^H \mathbf{h}_{lk}^l \big\} \big|^2 }{ \sum_{l'=1}^L \sum_{k'=1}^{K} p_{l'k'} \mathbb{E} \big\{ \big| \mathbf{v}_{lk}^H \mathbf{h}_{l'k'}^l \big|^2 \big\} -  p_{lk} \big| \mathbb{E} \big\{ \mathbf{v}_{lk}^H \mathbf{h}_{lk}^l \big\} \big|^2  + \sigma^2 \mathbb{E} \{ \|\mathbf{v}_{lk} \|^2 \}}.
\end{equation}
\vspace*{-0.1cm}
\hrulefill
\vspace*{-0.5cm}
\end{figure*}
The expectations in \eqref{eq:GeneralSINR} are computed over all the sources of randomness and \eqref{eq:AchievableRate} is an achievable rate since it is a lower bound on the channel capacity. This achievable rate can be computed numerically for any combining scheme, but with a high cost since many instantaneous channels need to be gathered such that several expectations can be numerically estimated.
\vspace*{-0.2cm}
\subsection{Uplink Spectral Efficiency Analysis}
\vspace*{-0.15cm}
If  maximum ratio (MR) combining is used by each BS, i.e.,$\big(\mathbf{v}_{lk} = \hat{\mathbf{h}}_{lk}^l\big),\forall l,k$, we obtain the closed-form expression for the uplink SE in \eqref{eq:AchievableRate} as shown by Theorem~\ref{Theorem:ClosedFormMR}.
\begin{theorem} \label{Theorem:ClosedFormMR}
When BS~$l$ uses the MR combing to decode the desired signal from user~$k$ in cell~$l$, the uplink SE obtained in \eqref{eq:AchievableRate} with the closed-form expression of the SINR value
\begin{equation} \label{eq:SINRlkMR}
\mathrm{SINR}_{lk} = \frac{ p_{lk} z_{lk}^l \Big| \mathrm{tr} \big(\mathbf{R}_{lk}^l \pmb{\Psi}_{lk}^l \mathbf{R}_{lk}^l \big) \Big|^2 }{ \mathsf{NI}_{lk} + \mathsf{CI}_{lk} + \mathsf{NO}_{lk} },
\end{equation}
where $\mathsf{NI}_{lk}, \mathsf{CI}_{lk},$ and $\mathsf{NO}_{lk}$ are respectively the non-coherent interference, coherent interference, and noise, which are 
\begin{align}
&\mathsf{NI}_{lk} = \sum_{l' =1 }^L \sum_{k'=1}^{K}  p_{l'k'} m_{l'k'}^l   \mathrm{tr} \big( \mathbf{R}_{lk}^l \pmb{\Psi}_{lk}^l \mathbf{R}_{lk}^l \mathbf{R}_{l'k'}^l  \big), \label{eq:NIlk} \\
&\mathsf{CI}_{lk} =  \sum_{(l',k') \in \mathcal{P}_{lk} \setminus (l,k)} p_{l'k'} z_{l'k'}^l \Big| \mathrm{tr} \big(  \mathbf{R}_{l'k'}^l \pmb{\Psi}_{lk}^l  \mathbf{R}_{lk}^l  \big)\Big|^2  + \notag \\
& \sum_{(l',k') \in \mathcal{P}_{lk}} p_{l'k'}     \frac{z_{l'k'}^l \mathrm{tr} \Big(\big(\widetilde{\mathbf{R}}_{l'k'}^l\big)^2 \Big)}{\big(d_{l'k'}^l S_{l'k'}^l\big)^2} \Big| \mathrm{tr} \big(  \mathbf{R}_{l'k'}^l \pmb{\Psi}_{lk}^l  \mathbf{R}_{lk}^l  \big)\Big|^2  \notag \\
& + \sum_{(l',k') \in \mathcal{P}_{lk}} p_{l'k'} z_{l'k'}^l  \frac{ \mathrm{tr} \Big(\big(\widetilde{\mathbf{R}}_{l'k'}^l \big)^2 \Big)}{\big(S_{l'k'}^l\big)^2} \times \notag \\
& \qquad \qquad \qquad \mathrm{tr} \left(  \mathbf{R}_{l'k'}^l \pmb{\Psi}_{lk}^l  \mathbf{R}_{lk}^l  \mathbf{R}_{l'k'}^l \mathbf{R}_{lk}^l \pmb{\Psi}_{lk}^l   \right), \label{eq:CIlk}\\
\mathsf{NO}_{lk} &=  \sigma^2 \hat{p}_{lk} \big(\beta_{lk}^l \big)^2 \big(d_{lk}^l \big)^2 \tau_p \mathrm{tr} \big(\mathbf{R}_{lk}^l \pmb{\Psi}_{lk}^l \mathbf{R}_{lk}^l \big),
\end{align}
with 
$m_{l'k'}^l = \beta_{l'k'}^l d_{l'k'}^l \hat{p}_{lk} \big(\beta_{lk}^l \big)^2 \big(d_{lk}^l \big)^2 \tau_p,$
and
$z_{l'k'}^l = \hat{p}_{l'k'}  \big(\beta_{l'k'}^l \big)^2 \big(d_{l'k'}^l \big)^2 \hat{p}_{lk} \big(\beta_{lk}^l \big)^2 \big( d_{lk}^l \big)^2 \tau_p^2,$ $\forall l',k',l.$
\end{theorem}
\begin{proof}
The proof is obtained by computing the expectations of non-Gaussian random variables in \eqref{eq:GeneralSINR}. The detailed proof is omitted due to space limitations.
\end{proof}
The SINR expression \eqref{eq:SINRlkMR} is explicitly influenced by many factors such as channel covariance matrices, the number of scatterers, pilot reuse, channel estimation quality, which are hidden in the general formulation \eqref{eq:GeneralSINR}. Specifically, the numerator of  \eqref{eq:SINRlkMR} shows the contribution of both channel estimation quality and covariance matrix of user~$k$ in cell~$l$.  Moreover, the effectiveness of the array gain is verified since the numerator scales up with the number of antennas. The first part in the denominator of \eqref{eq:SINRlkMR} demonstrates the degradation of the received signal quality due to non-coherent interference. The second part presents the contributions of coherent interference caused by reusing the pilot signals among the users. Unlike previous works with many scatterers \cite{Chien2018a}, this part also points out that a small number of scatterers have significant contributions to increase non-coherent interference. The last part in the denominator of \eqref{eq:SINRlkMR} represents additive noise effects.
\vspace*{-0.25cm}
\section{Uplink Total Data Energy Consumption Minimization} \label{Sec:Optimization}
\vspace*{-0.25cm}
This section expresses an uplink energy consumption minimization problem by assuming that user~$k$ in cell~$l$ requests a SE $\xi_{lk} > 0, \forall l,k,$ and has a maximum power $P_{\max,lk} >0$. Investigating this optimization problem, we  manifest the feasibility and infeasibility under the limited power budget.
\vspace*{-0.4cm}
\subsection{Problem Formulation}
\vspace*{-0.25cm}
The main goal of 5G-and-beyond systems is to provide the high SEs to all users with a minimal power consumption. In this paper, we formulate a total data energy optimization problem for the uplink data transmission as follows  
\begin{equation} \label{Prob:TotalPowerOptv1}
\begin{aligned}
& \underset{\{ p_{lk} \geq 0 \} }{\mathrm{minimize}}
&&   (\tau_c - \tau_p) \sum_{ l=1}^L \sum_{k=1}^K p_{lk} \\
& \,\,\mathrm{subject \,to}
& &  R_{lk} \geq \xi_{lk}, \forall l,k, \\
& & & p_{lk} \leq P_{\max,lk}, \forall l,k,
\end{aligned}
\end{equation}
where $ P_{\max,lk}$ is the maximum power level that user~$k$ in cell~$l$ can allocate to each data symbol. Problem~\eqref{Prob:TotalPowerOptv1} constrains on the rate requirement and limited power budget of each user. By setting $\nu_{lk} = 2^{\frac{\xi_{lk} \tau_c}{\tau_c - \tau_p}} - 1$ and removing the constant $\tau_c - \tau_p$ in the objective function, problem~\eqref{Prob:TotalPowerOptv1} is converted from the SE constraints into the equivalent SINR constraints as
\begin{equation} \label{Prob:TotalPowerOptv3}
\begin{aligned}
& \underset{\{ p_{lk} \geq 0 \} }{\mathrm{minimize}}
&&    \sum_{ l=1}^L \sum_{k=1}^K p_{lk} \\
& \,\,\mathrm{subject \,to}
& &   \frac{ p_{lk} z_{lk}^l \left| \mathrm{tr} \left(\mathbf{R}_{lk}^l \pmb{\Psi}_{lk}^l \mathbf{R}_{lk}^l \right) \right|^2 }{ \mathsf{NI}_{lk} + \mathsf{CI}_{lk} + \mathsf{NO}_{lk}} \geq \nu_{lk}, \forall l,k, \\
& & & p_{lk} \leq P_{\max,lk}, \forall l,k.
\end{aligned}
\end{equation}
We stress that problem~\eqref{Prob:TotalPowerOptv3} jointly optimizes the powers to satisfy the requested SINRs from all the users. This problem can be either feasible or infeasible for a given set of user locations and shadow fading realizations. 
\vspace*{-0.1cm}
\subsection{Feasible and Infeasible Problems} \label{SubSec:FeasInfeas}
\vspace*{-0.1cm}
When problem~\eqref{Prob:TotalPowerOptv3} has a non-empty feasible domain meaning that the network can simultaneously provide the required SEs to all the users conditioned on the power constraints. The global optimal solution to  problem~\eqref{Prob:TotalPowerOptv3} can be then found. Indeed, it is straightforward to show that \eqref{Prob:TotalPowerOptv3} is a linear program on standard form \cite{Boyd2004a}. We hence enable to solve \eqref{Prob:TotalPowerOptv3} to the global optimality in polynomial time, for instance, utilizing a general interior-point optimization toolbox as CVX \cite{cvx2015}. It should be noticed that all the $KL$ users will spend non-zero data powers at the global optimum when \eqref{Prob:TotalPowerOptv3} is feasible owning to the  non-zero SE requirements.

There may be a situation that all the users cannot be simultaneously served by the SE requirements. Only one unfortunate user served with a lower SE suffices to create an empty feasible domain for the total transmit power optimization problem. Alternatively, problem~\eqref{Prob:TotalPowerOptv3} lacks a feasible solution \cite[Section 4.1]{Boyd2004a}. The unsatisfied SE is caused by high mutual interference in cellular networks and/or extreme locations as the cell edge leading to some users having a weak channel. Moreover, a user may require a too high SE and the system cannot provide this service even spending maximum data power. 
A feasible solution might still exist for most of the users with the required SEs by expecting that only one or a few users are \textit{unsatisfied}. It is sufficient to remove or reduce the required SEs of those unsatisfied users to convert an infeasible problem to a feasible one. This paper develops the power allocation strategies to handle such infeasible instances by allowing the corresponding SINR constraints to be violated.
\vspace*{-0.2cm}
\section{Congestion solution based on alternating optimization} \label{Sec:Solutions}
\vspace*{-0.15cm}
This section proposes the two algorithms attaining a fixed-point solution to problem~\eqref{Prob:TotalPowerOptv3} with either empty or non-empty feasible domain.
\vspace*{-0.2cm}
\subsection{Spending Maximum Transmit Power on Unsatisfied Users}
\vspace*{-0.15cm}
For the glorification of simplification in comprehension, problem~\eqref{Prob:TotalPowerOptv3} with a non-empty feasible domain is first considered. We stack all the data powers into a vector $\mathbf{p} = [p_{11}, \ldots, p_{LK}]^T \in \mathbb{R}_{+}^{LK}$, then the SINR constraint of user~$k$ in cell~$l$ is reformulated as
$p_{lk} \geq I_{lk} (\mathbf{p})$,
where $I_{lk} (\mathbf{p})$ is so-called a standard interference function, which is
\begin{equation} \label{eq:Ilk}
I_{lk} (\mathbf{p}) = \frac{\nu_{lk} \mathsf{NI}_{lk} (\mathbf{p}) + \nu_{lk} \mathsf{CI}_{lk} (\mathbf{p}) + \nu_{lk} \mathsf{NO}_{lk} }{ z_{lk}^l \left| \mathrm{tr} \left(\mathbf{R}_{lk}^l \pmb{\Psi}_{lk}^l \mathbf{R}_{lk}^l \right) \right|^2 }.
\end{equation}
In \eqref{eq:Ilk}, the detailed expressions of $\mathsf{NI}_{lk} (\mathbf{p})$ and $\mathsf{CI}_{lk} (\mathbf{p})$ have been already expressed in \eqref{eq:NIlk} and \eqref{eq:CIlk}, but we here emphasize them as the functions of data power variables stacked in $\mathbf{p}$. We now introduce the definition of a standard interference function for which an algorithm to obtain a fixed point solution is proposed.
\begin{definition}[Standard interference function] \label{Def:TypeI}
A function $I(\mathbf{z})$ is a standard interference function for all $\mathbf{z} \succeq \mathbf{0}$, if the following properties hold:\footnote{The notation $\mathbf{z} \succeq \mathbf{z}'$ indicates element-wise inequality $z_n \geq z_n', \forall n = 1, \ldots, KL$ with $z_n$ and $z_n'$ being the elements of $\mathbf{z}$ and $\mathbf{z}'$, respectively.}
$a)$ Positivity $I(\mathbf{z}) >0, \forall \mathbf{z} >0$. $b)$ Monotonicity $I(\mathbf{z})  \geq I(\mathbf{z}')$ if $\mathbf{z} \succeq \mathbf{z}'$. $c)$ Scalability: $\alpha I(\mathbf{z}) > I(\alpha \mathbf{z}), \forall \alpha > 1,$ for all scalar $\alpha > 1$.
\end{definition}
The positivity property is because of the inherent mutual interference and thermal noise in the system, which implies a non-zero value data powers when users request non-zero SEs. The monotonicity property ensures that we can scale up or down \eqref{eq:Ilk} by adjusting the data powers. Finally, the scalability property suggests a method to uniformly scale down the data power coefficient of user~$k$ in cell~$l$ 
at each iteration by utilizing a positive constant $\alpha$. We now construct a policy to update the data power of every user~$k$ in cell~$l$ in Theorem~\ref{Theorem:Alg1}.
\begin{theorem}\label{Theorem:Alg1}
By assuming that the feasible domain is non-empty and $0 \leq I_{lk} (\mathbf{p}) \leq P_{\max,lk}^2$ always holds for all $\mathbf{p}$ in the feasible domain. For the initial values of data powers $p_{lk} (0) = P_{\max,lk}, \forall l,k$, there exist data powers for which each interference function $I_{lk} (\mathbf{p})$ is non-increasing along iterations and converges to a fixed point. Particularly, the data power of user~$k$ in cell~$l$, denoted by $p_{lk}(n)$, can be updated at iteration~$n$ as
\begin{equation} \label{eq:plkn}
p_{lk}(n) = I_{lk} (\mathbf{p}(n-1)), \forall l,k.
\end{equation}
\end{theorem}
\begin{proof}
The proof is to testify $I_{lk} (\mathbf{p}),\forall l,k,$ defined in \eqref{eq:Ilk} being standard interference, so the updated power policy \eqref{eq:plkn} ensures that this iterative approach converges to a fixed point. The detailed proof is omitted due to space limitations. 
\end{proof}
Every user has its own standard interference function satisfying the three fundamental properties in Definition~\ref{Def:TypeI} and utilizing it to update the data power as in \eqref{eq:plkn}. The analysis in Theorem~\ref{Theorem:Alg1} is based on the assumption that problem~\eqref{Prob:TotalPowerOptv3} has the global optimum for which all users are served with their required SEs. The power constraints in \eqref{Prob:TotalPowerOptv3} ($p_{lk} \leq P_{\max, lk}, \forall l,k$) are tackled by the fact if $I_{lk}(n-1) > P_{\max,lk}$, then the congestion issue appears and leads to an obvious selection $p_{lk} (n) =  P_{\max,lk}$. We therefore define the constrained standard interference function used at iteration~$n-1$ as
\begin{equation} \label{eq:Ihatlk}
\hat{I}_{lk}(\mathbf{p}(n-1)) = \min \left( I_{lk}(\mathbf{p}(n-1)), P_{\max,lk}   \right).
\end{equation}
For a cellular massive MIMO system with the power budget constraints and the initial data power vector $\mathbf{p} (0)$ with the entries $p_{lk}(0) = P_{\max,lk}, \forall l,k,$ iteration~$n$ updates the data power of user~$k$ in cell~$l$ as
\begin{equation} \label{eq:plkv1}
p_{lk}(n) = \hat{I}_{lk}(\mathbf{p}(n-1)).
\end{equation}
Combining \eqref{eq:Ihatlk} and \eqref{eq:plkv1}, we observe that if $\hat{I}_{lk}(\mathbf{p}(n-1)) = P_{\max,lk}$, the update $p_{lk}(n) = P_{\max,lk}$ maintains the non-increasing objective function of problem~\eqref{Prob:TotalPowerOptv3}. Otherwise, it holds that $\hat{I}_{lk}(\mathbf{p} (n-1)) = I_{lk}(\mathbf{p} (n-1))$, and hence user~$k$ in cell~$l$ consumes less power than the maximum. This procedure will be applied to all the $KL$ users, which results in an alternating approach summarized in Algorithm~\ref{Algorithm1}. Notice that, when users cannot be served by the
required SEs, one still lets them utilize the maximum power. This policy aims at maximizing the SE of a particular user, however producing more mutual interference to the other users.

\begin{algorithm}[t]
	\caption{Data power allocation to problem~\eqref{Prob:TotalPowerOptv3} by spending maximum transmit power on unsatisfied users} \label{Algorithm1}
	\textbf{Input}:  Define $P_{\max,lk}, \forall l,k$; Select  $p_{lk}(0) = P_{\max,lk}, \forall l,k$; Compute the total data power $P_{\mathrm{tot}}(0) = \sum_{l=1}^L \sum_{k=1}^K p_{lk}(0)$; Set initial value $n=1$ and tolerance $\epsilon$.
	\begin{itemize}
		\item[1.] User~$k$ in cell~$l$ computes the standard interference function $	{I}_{lk} \left(\mathbf{p} (n-1) \right)$ using \eqref{eq:Ilk}.
		\item[2.] If ${I}_{lk} \left(\mathbf{p} (n-1) \right) > P_{\max,lk}$, update $p_{lk}(n) = P_{\max,lk}$. Otherwise, update $ p_{lk}(n) = 	{I}_{lk} \left(\mathbf{p} (n-1) \right) $.
		\item[3.] Repeat Steps $1,2$ with other users, then compute the ratio \fontsize{9}{9}{$\gamma (n) =$ $| P_{\mathrm{tot}}(n) - P_{\mathrm{tot}}(n-1) | /  P_{\mathrm{tot}}(n-1)$}.
		\item[4.] If $\gamma_l (n) \leq \epsilon$ $\rightarrow$ Set $p_{lk}^{\ast} = p_{lk}(n),\forall l,k,$ and Stop. Otherwise, set $n= n+1$ and go to Step $1$.
	\end{itemize}
	\textbf{Output}: A fixed point $p_{lk}^{\ast}$, $\forall l,k$. \vspace*{-0.05cm}
\end{algorithm}
\vspace*{-0.2cm}
\subsection{Softly Removing Unsatisfied Users}
\vspace*{-0.1cm}
Instead of allowing potential unsatisfied users to spend full data power,  one can reduce their power with the goal to degrade mutual interference to the others. This policy might ameliorate the number of satisfied users in the entire network. At first, every user improves the transmission quality by spending more power to each data symbol. This target can be achieved by, for example, simply constructing the standard inference functions as in the previous subsection. If at the limited power budget, the required SE cannot be achieved, unsatisfied users will reduce data power. We then mathematically suggest an update of the data powers along iterations as follows.
\begin{theorem} \label{theorem:Standardfunction}
From initial values $p_{lk}(0) = P_{\max,lk}, \forall l,k,$ if the data power of user~$k$ in cell~$l$ is updated at iteration~$n$ as
\begin{equation} \label{eq:UpdatedPowerv1}
p_{lk}(n)  = \begin{cases}
I_{lk} \left( \mathbf{p} (n-1) \right),& \mbox{if } I_{lk} \left( \mathbf{p} (n-1) \right) \leq P_{\max,lk}, \\
\frac{P_{\max,lk}^2}{I_{lk} \left(\mathbf{p} (n-1) \right)}, & \mbox{if } I_{lk} \left( \mathbf{p} (n-1) \right) > P_{\max,lk}, 
\end{cases}
\end{equation}
then the iterative approach converges to a fixed point.
\end{theorem}
\begin{proof}
The proof is first to confirm that the updated power policy in \eqref{eq:UpdatedPowerv1} follows a so-called two-sided  function and the convergence is then established. The detailed proof is omitted due to space limitations.
\end{proof} 
This theorem provides a procedure to minimize the total transmit power in the network and coping with the congestion issue. If $I_{lk}(\mathbf{p}(n-1))$ is less than $P_{\max, lk}$, then the data power of user~$k$ in cell~$l$ is updated based on \eqref{eq:plkn}, same as Algorithm~\ref{Algorithm1}. The main distinction is to prevent any unsatisfied user from transmitting full power whenever the congestion issue happens, i.e. $I_{lk} (\mathbf{p}(n-1)) > P_{\max,lk}$. By doing this power update, the mutual interference from this unsatisfied user to the others should be reduced, and hence there is chance for the remaining users to get their required SEs. The proposed optimization approach is summarized in Algorithm~\ref{Algorithm2} and its properties are stated in Remark~\ref{Remarlk:Property}.
\begin{remark}\label{Remarlk:Property}
The proposed algorithms work in both feasible and infeasible domain such that a fixed point to problem~\eqref{Prob:TotalPowerOptv3} can be obtained. For realizations of user locations that result in feasible domains, the fixed point obtained by those algorithms is the global optimum. The main difference between the two algorithms appears whenever the congestion issue happens: While Algorithm~\ref{Algorithm1} allocates the maximum data power to users when their SINR constraints are not satisfied, Algorithm~\ref{Algorithm2} reduces the data power.  Thus, the fixed point obtained by each algorithm may be different from each other.  
\end{remark}
\begin{algorithm}[t]
	\caption{Data power allocation to problem~\eqref{Prob:TotalPowerOptv3} by softly removing unsatisfied users} \label{Algorithm2}
	\textbf{Input}:  Define $P_{\max,lk}, \forall l,k$; Select $p_{lk}(0) = P_{\max,lk}, \forall l,k$; Compute the total data power $P_{\mathrm{tot}}(0) = \sum_{l=1}^L \sum_{k=1}^K p_{lk}(0)$; Set initial value $n=1$ and tolerance $\epsilon$.
	\begin{itemize}
		\item[1.] User~$k$ in cell~$l$ computes the standard interference function $	{I}_{lk} \left(\mathbf{p} (n-1) \right)$ using \eqref{eq:Ilk}.
		\item[2.] If ${I}_{lk} \left(\mathbf{p} (n-1) \right) \leq P_{\max,lk}$,  $ p_{lk}(n) = {I}_{lk} \left(\mathbf{p} (n-1) \right) $. Otherwise,  $p_{lk}(n) =  P_{\max,lk}^2 / {I}_{lk} \left(\mathbf{p} (n-1) \right)$.
		\item[3.] Repeat Steps $1,2$ with other users, then compute the ratio \fontsize{9}{9}{$\gamma (n) =$ $| P_{\mathrm{tot}}(n) - P_{\mathrm{tot}} (n-1) | /   P_{\mathrm{tot}}(n-1)$}.
		\item[4.] If $\gamma_l (n) \leq \epsilon$ $\rightarrow$ Set $p_{lk}^{\ast} = p_{lk}(n),\forall l,k,$ and Stop. Otherwise, set $n= n+1$ and go to Step $1$.
	\end{itemize}
	\textbf{Output}: A fixed point $ p_{lk}^{\ast}$, $\forall l,k$. 
\end{algorithm}

\begin{figure}[t]
	\centering
	\includegraphics[trim=3.0cm 7cm 3.5cm 8.5cm, clip=true, width=3in]{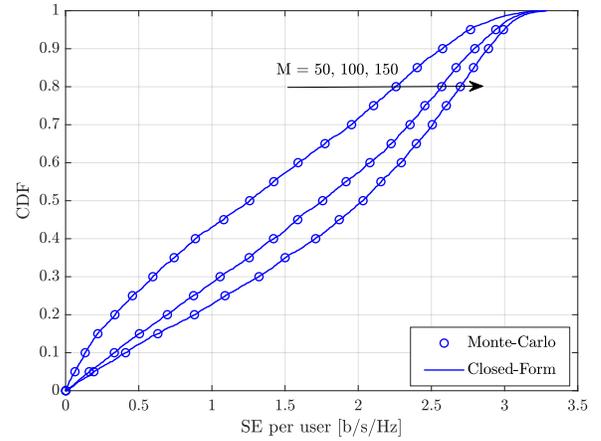} \vspace*{-0.5cm}
	\caption{The CDF of the uplink SE per user [b/s/Hz] with Monte-Carlo simulation and closed-form expression. }
	\label{FigMonteCarloClosedForm}
	\vspace*{-0.4cm}
\end{figure}
\begin{figure}[t]
	\centering
	\includegraphics[trim=3.0cm 7cm 3.5cm 8.5cm, clip=true, width=3in]{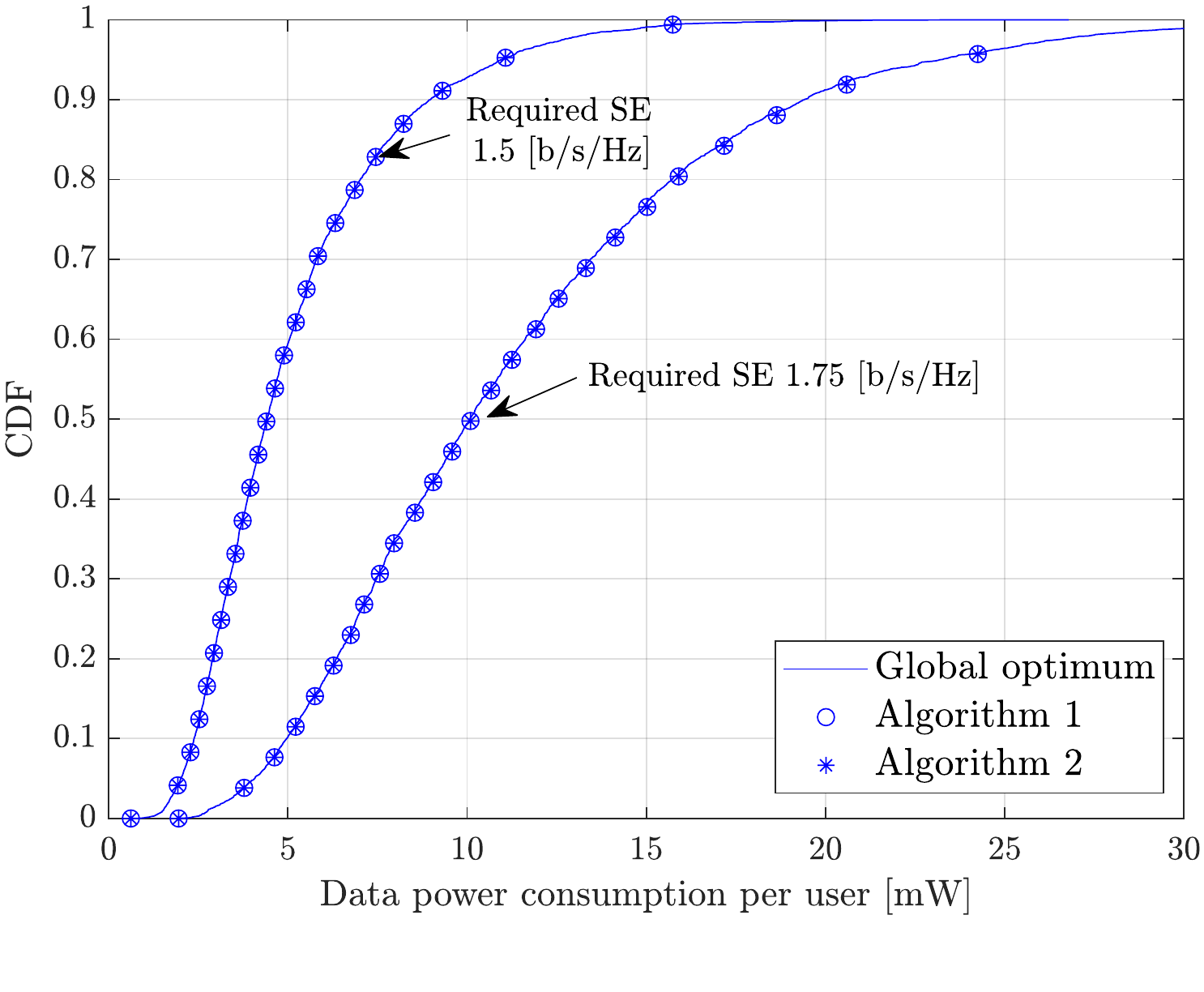} \vspace*{-0.5cm}
	\caption{The CDF of the power consumption per user [mW] for feasible systems with the different required SEs at the users and $M=100$. }
	\label{FigPowerFeas}
	\vspace*{-0.4cm}
\end{figure}

\begin{figure}[t]
	\centering
	\includegraphics[trim=2.6cm 7cm 3.5cm 8.5cm, clip=true, width=3in]{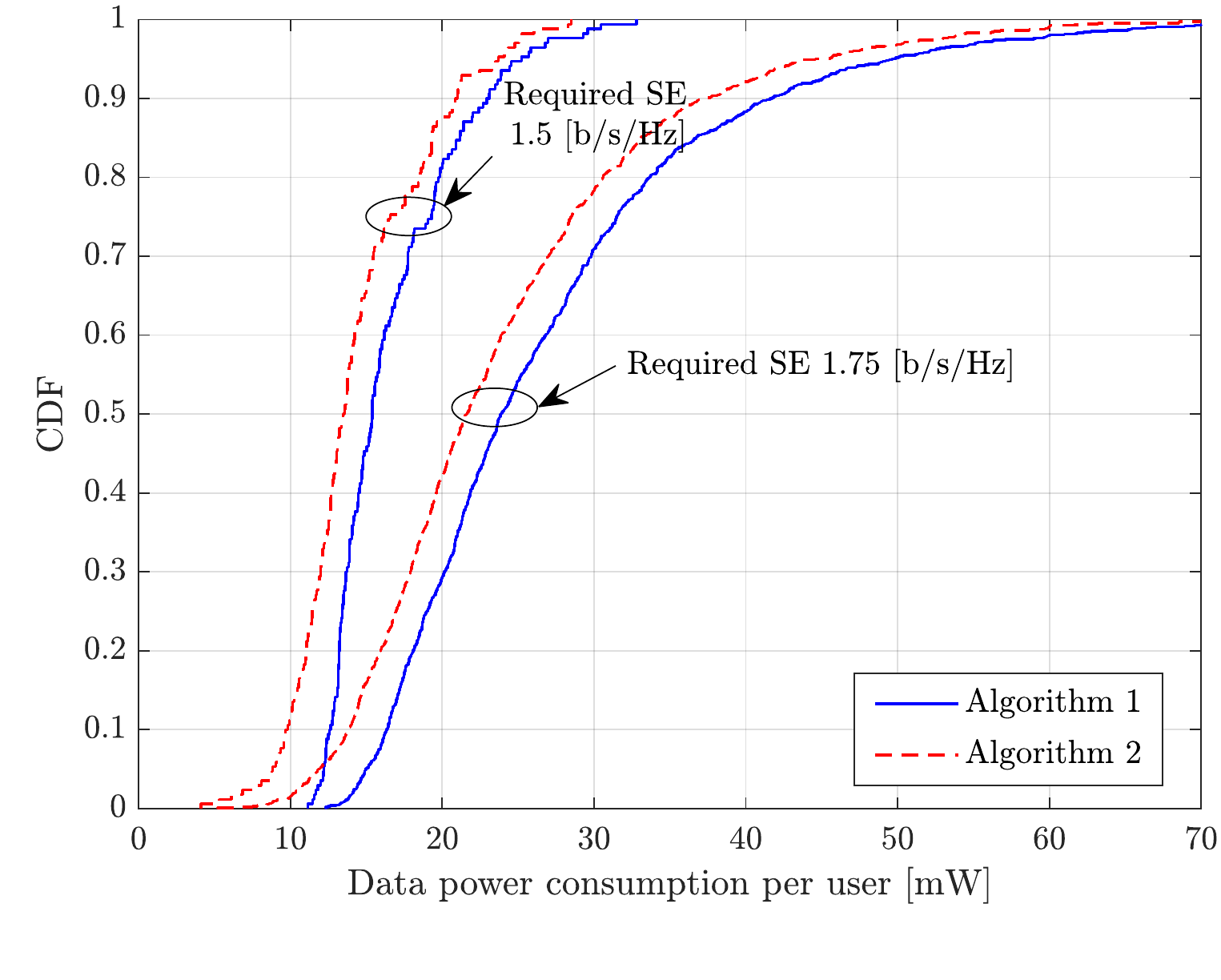} \vspace*{-0.5cm}
	\caption{The CDF of the power consumption per user [mW] for infeasible systems with the different required SEs at the users and $M=100$.}
	\label{FigPowerInfeas}
	\vspace*{-0.4cm}
\end{figure}

\begin{figure}[t]
	\centering
	\includegraphics[trim=3.0cm 7cm 3.5cm 8.5cm, clip=true, width=3in]{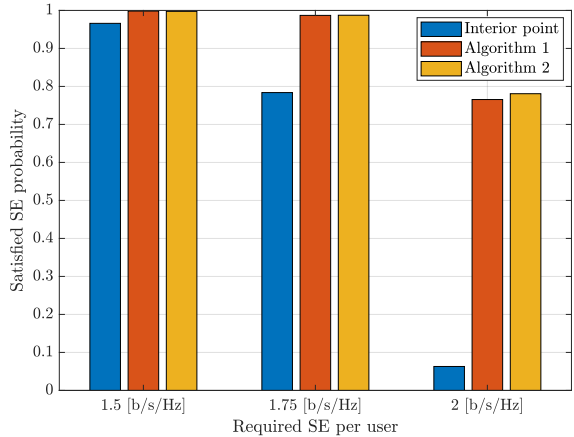} \vspace*{-0.5cm}
	\caption{The satisfied SE probability versus the different required SE per user for a system with $M=100$.}
	\label{FigProb}
	\vspace*{-0.4cm}
\end{figure}
\vspace*{-0.2cm}
\section{Numerical Results} \label{Sec:NumericalResults}
\vspace*{-0.1cm}
A massive MIMO system is considered with $L=4$ square cells in a $1$~km$^2$ area, each serving $K=5$ users. All the users are uniformly distributed within its cell with the distance to the BS no less than $35$~m. Each coherence block has $\tau_c = 200$ symbols and $\tau_p = 5$ orthogonal pilot signals with  $\hat{p}_{lk} = P_{\max,lk}= 200$~mW, $\forall l,k$. The users with same index in all cells share the same pilot sequence. The system bandwidth is $20$~MHz and the noise variance is $-96$~dBm with the noise figure of $5$~dB. The large-scale fading coefficient [dB] of user~$k$ in cell~$l$ and BS~$l'$ is modeled based on the 3GPP LTE specifications \cite{LTE2010a} as
$\beta_{lk}^{l'} = -148.1 - 37.6 \log_{10} ( d_{lk}^{l'} / 1 \mbox{km}) + z_{lk}^{l'}$,
where $d_{lk}^{l'} > 35$~m is the distance between user~$k$ in cell~$l$ and BS~$l'$; $z_{lk}^{l'}$ is the shadow fading coefficient following a Gaussian distribution with zero mean and standard deviation $7$~dB. There are $21$ scatterers per communication link. The covariance matrices are computed by using \cite{van2016multi}. In the proposed algorithms, $\epsilon = 0.001$. For feasible systems, the global optimum obtained by interior point methods \cite{van2020power} are included for comparison.

Figure~\ref{FigMonteCarloClosedForm} shows the cumulative distribution function (CDF) of SE per user [b/s/Hz] to verify the correctness of the closed-form expression of the uplink SE for each user. All users spend full power for the data transmission. Particularly, the closed-form expression result matches very well Monte-Carlo simulation result for all the considered number of BS antennas. Fig.~\ref{FigMonteCarloClosedForm} also illustrates the SE per user gets better when each BS is equipped with more antennas. Each user can be served by a data rate increasing from $1.3$ to $1.8$ [b/s/Hz] on the average if the number of BS antennas increases from $50$ to $150$.

The CDF of the data power consumption [mW] consumed by each user is shown in Fig.~\ref{FigPowerFeas} for feasible systems. The proposed algorithms provide a unique fixed point that is the global optimum as what has been obtained by the interior-point methods. Furthermore, data power escalates when users require higher SEs. With a required SE of $1.5$ [b/s/Hz], each user only spends $5.2$ mW on average. But it drastically grows to $11.4$ mW if the required SE is $1.75$ [b/s/Hz]. 

Figure~\ref{FigPowerInfeas} displays the  CDF of the data power consumption [mW] per user for infeasible systems which is the main interest of this paper when working with multiple access in massive MIMO, since there is no global optimum to obtain. All the users consume non-zero powers at the fixed points identified Algorithms~\ref{Algorithm1} and \ref{Algorithm2}. The data power consumption per user obtained by Algorithm~\ref{Algorithm1} are $12.3\%$ and $15.1\%$ higher than the ones by Algorithm~\ref{Algorithm2} for the given required SEs.

Figure~\ref{FigProb} plots the satisfied SE probability defined as the fraction of the number of large-scale fading realizations in which the users can be served by the required SEs. If each user requires an SE $1.5$ [b/s/Hz], all the benchmarks provide an overwhelming satisfied SE probability. The interior-point methods perform worse with higher SE requirements, especially only $6.3\%$ users satisfied the required SE $2$ [b/s/Hz]. In contrast, the proposed algorithms still offer a satisfied SE probability of more than $75\%$.
\vspace*{-0.15cm}
\section{Conclusion} \label{Sec:Conclusion}
\vspace*{-0.1cm}
This paper has analyzed the system performance of massive MIMO systems with an arbitrary number of BS antennas, users, and scatterers by utilizing the double scattering channel model, rather than the asymptotic regime as in previous works. The closed-form expression of the uplink SE per user was computed in closed form. We proposed two algorithms to handle effectively the congestion issue that often happens since multiple users are simultaneously connecting to the network and sharing the same time and frequency resources.
 \vspace*{-0.5cm}
\bibliographystyle{IEEEtran}
\bibliography{IEEEabrv,refs}
\end{document}